\documentclass[a4paper,USenglish,cleveref, autoref]{lipics-v2021}

\title{Filling Crosswords is Very Hard} 

\titlerunning{Filling Crosswords is Very Hard}

\author{Laurent Gourv{\`e}s}{Universit\'e Paris-Dauphine, Universit\'e PSL, CNRS, LAMSADE, 75016, Paris, France}{laurent.gourves@dauphine.fr}{}{}
\author{Ararat Harutyunyan}{Universit\'e Paris-Dauphine, Universit\'e PSL, CNRS, LAMSADE, 75016, Paris, France}
{ararat.harutyunyan@dauphine.fr}{}{}
\author{Michael Lampis}{Universit\'e Paris-Dauphine, Universit\'e PSL, CNRS, LAMSADE, 75016, Paris, France}
{michail.lampis@dauphine.fr}{}{}
\author{Nikolaos Melissinos}{Universit\'e Paris-Dauphine, Universit\'e PSL, CNRS, LAMSADE, 75016, Paris, France}
{nikolaos.melissinos@dauphine.eu}{}{}

\authorrunning{Gourv{\`e}s, Harutyunyan, Lampis, Melissinos}

\Copyright{Anonymous}

\ccsdesc[500]{Theory of Computation $\rightarrow$ Design and Analysis of Algorithms $\rightarrow$ Parameterized Complexity and Exact Algorithms}
\ccsdesc[500]{Theory of Computation $\rightarrow$ Design and Analysis of Algorithms $\rightarrow$ Approximation Algorithms}

\keywords{Crossword Puzzle, Treewidth, ETH}

\category{}

\relatedversion{}

\supplement{}

\funding{}

\acknowledgements{We thank Dominique Taton for communicating the crossword puzzle to us.}

\nolinenumbers 

\EventEditors{John Q. Open and Joan R. Access}
\EventNoEds{2}
\EventLongTitle{42nd Conference on Very Important Topics (CVIT 2016)}
\EventShortTitle{CVIT 2016}
\EventAcronym{CVIT}
\EventYear{2016}
\EventDate{December 24--27, 2016}
\EventLocation{Little Whinging, United Kingdom}
\EventLogo{}
\SeriesVolume{42}
\ArticleNo{23}

\usepackage[noend]{algpseudocode}
\usepackage[Alg.]{algorithm}
\usepackage{mathtools}
\usepackage[svgnames]{xcolor}

\usepackage{todonotes}
\usepackage{multirow}

\newcommand{\I}{\mathcal{I}}

\newtheorem{property}{Property}

\newcommand{\w}{d}

\def\D{{\cal D}}
\def\L{{\cal L}}
\def\I{{\cal I}}
\def\T{{\cal T}}

\def\tw{\texttt{tw}}

\newcommand{\commentLaurent}[1]{\textcolor{red}{Laurent: #1}}

\def\CPD{\textsc{CP-Dec}}
\def\CPO{\textsc{CP-Opt}}
\def\CPDlong{\textsc{Crossword Puzzle Decision}}
\def\CPOlong{\textsc{Crossword Puzzle Optimization}}

\tikzstyle{vertex}=[circle, draw, inner sep=1.2pt, minimum width=4pt, minimum size=0.4cm]
\tikzstyle{vertex2}=[circle, draw, inner sep=0pt, minimum width=4pt, minimum size=0.15cm]
\usetikzlibrary{decorations,decorations.pathmorphing,decorations.pathreplacing,fit,positioning,arrows}

\begin{document}

\maketitle

\begin{abstract}

We revisit a classical crossword filling puzzle which already appeared in Garey\&Jonhson's book. We are given a grid with $n$ vertical and horizontal slots and a dictionary with $m$ words and are asked to place words from the dictionary in the slots so that shared cells are consistent. We attempt to pinpoint the source of intractability of this problem by carefully taking into account the structure of the grid graph, which contains a vertex for each slot and an edge if two slots intersect. Our main approach is to consider the case where this graph has a tree-like structure. Unfortunately, if we impose the common rule that words cannot be reused, we discover that the problem remains NP-hard under very severe structural restrictions, namely, if the grid graph is a union of stars and the alphabet has size $2$, or the grid graph is a matching (so the crossword is a collection of disjoint crosses) and the alphabet has size $3$. The problem does become slightly more tractable if word reuse is allowed, as we obtain an $m^{\tw}$ algorithm in this case, where $\tw$ is the treewidth of the grid graph. However, even in this case, we show that our algorithm cannot be improved to obtain fixed-parameter tractability. More strongly, we show that under the ETH the problem cannot be solved in time $m^{o(k)}$, where $k$ is the number of horizontal slots of the instance (which trivially bounds $\tw$).

Motivated by these mostly negative results, we also consider the much more restricted case where the problem is parameterized by the number of slots $n$. Here, we show that the problem does become FPT (if the alphabet has constant size), but the parameter dependence is exponential in $n^2$. We show that this dependence is also justified: the existence of an algorithm with running time $2^{o(n^2)}$, even for binary alphabet, would contradict the randomized ETH. Finally, we consider an optimization version of the problem, where we seek to place as many words on the grid as possible. Here it is easy to obtain a $\frac{1}{2}$-approximation, even on weighted instances, simply by considering only horizontal or only vertical slots. We show that this trivial algorithm is also likely to be optimal, as obtaining a better approximation ratio in polynomial time would contradict the Unique Games Conjecture. The latter two results apply whether word reuse is allowed or not.

\end{abstract}

\section{Introduction}

Crossword puzzles are one-player games where the goal is to fill a (traditionally two-dimensional) grid with words. Since their first appearance more than 100 years ago, crossword puzzles have rapidly become popular. Nowadays, they can be found in many newspapers and magazines around the world like the {\em New York Times} in the USA, or {\em Le Figaro} in France. 
Besides their obvious recreational interest, crossword puzzles are valued tools in education \cite{CC83} and medicine. In particular, crossword puzzles participation seems to delay the onset of  memory decline \cite{PHDBLV11}.
They are also helpful for developing and testing computational techniques; see for example \cite{Rosin11}. In fact, both the design and the completion of a puzzle are challenging.  In this article, we are interested in the task of solving a specific type of crossword puzzle.  

There are different kinds of crossword puzzles. In the most famous ones, some clues are given together with the place  where the answers should be located. A solution contains words that must be consistent with the given clues, and the intersecting pairs of words are constrained to agree on the letter they share. {\em Fill-in} crossword puzzles 
do not come with clues. 
Given a list of words and a grid in which some slots are identified, the objective is to fill all the slots with the given words. The list of words is typically succinct and provided explicitly.

In a variant of fill-in crossword puzzle currently proposed in a French TV magazine \cite{TM}, 
one has to find up to 14 words and place them in a grid (the grid is the same for every instance, see Figure \ref{fig1} for an illustration). The words are not explicitly listed but they must be {\em valid} (for instance, belong to the French language). In an instance of the game, some specified letters have a positive weight; the other letters have weight zero. 
The objective is to find a solution whose weight -- defined as the total sum of the letters written in the grid -- is at least a given threshold.

\begin{figure}
\center
\begin{tikzpicture}[scale=0.3]

\fill[color=black, line width=0pt] (0,0)--(2,0)--(2,5)--(3,5)--(3,4)--(4,4)--(4,7)--(1,7)--(1,1)--(0,1)--(0,0);

\fill[color=black, line width=0pt] (3,1)--(6,1)--(6,5)--(5,5)--(5,2)--(4,2)--(4,3)--(3,3)--(3,1);

\fill[color=black, line width=0pt] (5,6)--(6,6)--(6,8)--(5,8)--(5,6);

\fill[color=black, line width=0pt] (6,0)--(7,0)--(7,1)--(6,1)--(6,0);

\fill[color=black, line width=0pt] (7,6)--(8,6)--(8,8)--(7,8)--(7,6);

\fill[color=black, line width=0pt] (10,1)--(7,1)--(7,5)--(8,5)--(8,2)--(9,2)--(9,3)--(10,3)--(10,1);

\fill[color=white, line width=0pt] (13,0)--(11,0)--(11,5)--(10,5)--(10,4)--(9,4)--(9,7)--(12,7)--(12,1)--(13,1)--(13,0);

\fill[color=black, line width=0pt] (13,0)--(11,0)--(11,5)--(10,5)--(10,4)--(9,4)--(9,7)--(12,7)--(12,1)--(13,1)--(13,0);



\foreach \k in {0,...,8}
{
\draw (0,\k)--(13,\k);
}

\foreach \k in {0,...,13}
{
\draw (\k,0)--(\k,8);
}

\draw[thick] (0,0)--(13,0)--(13,8)--(0,8)--(0,0);

\end{tikzpicture}
\caption{Place valid words in this grid. In a possible instance, letters S, U, I, V, R, E, and T have weight 7, 5, 4, 2, 6, 1, and 3, respectively. 
Any other letter has null weight. 
Try to obtain at least 330 points. \label{fig1}}
\end{figure}
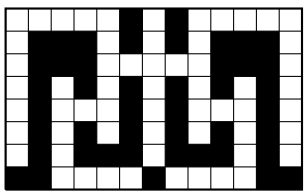

The present work deals with a theoretical study of this fill-in crossword puzzle (the grid is not limited to the one of Figure \ref{fig1}). We are mainly interested in two problems: Can the grid be entirely completed? How can the weight of a solution be maximized? Hereafter, these problems are 
called {\CPDlong} and {\CPOlong} 
({\CPD} and {\CPO} in short), respectively.


{\CPD} is not new; see  GP14 in \cite{GJ79}. 
The proof of NP-completeness is credited to a personal communication with Lewis and Papadimitriou. 
Thereafter, an alternative NP-completeness proof appeared in \cite{EHRS12} (see also \cite{LMS15}).
Other articles on crossword puzzles exist and they 
are mostly empirically validated 
techniques coming from 
Artificial Intelligence and Machine Learning; see for example 
\cite{GFHT90,MG97,LKS99,AB08,Rosin11,RDMG12} an references therein.  

\subparagraph*{Our Results} Our goal in this paper is to pinpoint the relevant structural parameters that make filling crossword puzzles intractable. We begin by examining the structure of the given grid. It is natural to think that, if the structure of the grid is tree-like, then the problem should become easier, as the vast majority of problems are tractable on graphs of small treewidth. We only partially confirm this intuition: by taking into account the structure of a graph that encodes the intersections between slots (the grid-graph) we show in Section~\ref{sec:tree:like} that {\CPO} can be solved in polynomial time on instances of constant treewidth. However, our algorithm is not fixed-parameter tractable and, as we show, this cannot be avoided, even if one considers the much more restricted case where the problem is parameterized by the number of horizontal slots, which trivially bounds the grid-graph's treewidth (Theorem~\ref{thm:byrows}). More devastatingly, we show that if we also impose the natural rule that words cannot be reused, the problem already becomes NP-hard when the grid graph is a matching for alphabets of size 3 (Theorem~\ref{thm:matching:3letters}), or a union of stars for a binary alphabet (Theorem~\ref{thm:forest:2letters}). Hence, a tree-like structure does not seem to be of much help in rendering crosswords tractable.

We then go on to consider {\CPO} parameterized by the total number of slots $n$. This is arguably a very natural parameterization of the problem, as in real-life crosswords, the size of the grid can be expected to be significantly smaller than the size of the dictionary. We show that in this case the problem does become fixed-parameter tractable (Corollary~\ref{corollary:algslots}), but the running time of our algorithm is exponential in $n^2$. Our main result is to show that this disappointing dependence is likely to be best possible: even for a binary alphabet, an algorithm solving {\CPD} in time $2^{o(n^2)}$ would contradict the randomized ETH (Theorem \ref{thm:all:slots:rand:eth}). Note that all our positive results up to this point work for the more general {\CPO}, while our hardness results apply to {\CPD}.

Finally, in Section \ref{sec:approximability} we consider the approximability of {\CPO}. Here, it is easy to obtain a $\frac{1}{2}$-approximation by only considering horizontal or vertical slots. We are only able to slightly improve upon this, giving a polynomial-time algorithm with ratio $\frac{1}{2}+O(\frac{1}{n})$. Our main result in this direction is to show that this is essentially best possible: obtaining an algorithm with ratio $\frac{1}{2}+\epsilon$ would falsify the Unique Games Conjecture (Theorem \ref{thm:UGC}).

\section{Problem Statement and Preliminaries} \label{sec:prel}

We are given a dictionary $\D=\{d_1,\ldots d_m\}$ whose words are constructed on an alphabet $\L=\{l_1,\ldots l_\ell\}$, and a two-dimensional grid consisting of horizontal and vertical slots. A slot is composed of consecutive cells. 
Horizontal slots do not intersect each other; the same goes for vertical slots. However horizontal slots can intersect vertical slots.  A cell is {\em shared} if it lies at the intersection of two slots. 
Unless specifically stated,  $n$, $m$ and $\ell$ denote the total number of slots, the size of $\D$, and the size of $\L$, respectively. Finally, let us mention that we consider only instances where the alphabet is of constant size, i.e., $\ell = O(1)$.


In a feasible solution, each slot $S$ receives either a word of $\D$ of length $|S|$, or nothing (we sometimes say that a slot receiving nothing gets an {\em empty word}). Each cell gets at most one letter, and the words assigned to two intersecting slots must agree on the letter placed in the shared cell. 
All filled horizontal slots get words written from left to right (across) 
while all vertical slots get words written from top to bottom (down).


There is a weight function $w: \L \rightarrow \mathbb{N}$. 
The weight of a solution is the total sum of the weights of the 
letters placed in the grid. 
Observe that, for a given solution, the total weight of all filled-in words is not the same as the weight of this solution as, in the latter, the letters of the shared cells are counted only once.

The two main problems studied in this article are the following. 
Given a grid, a dictionary $\D$ on alphabet $\L$, and a weight function $w: \L \rightarrow \mathbb{N}$, the objective of 
{\CPOlong} ({\CPO} in short) is to find a feasible solution of maximum weight. 
Given a grid and a dictionary $\D$ on alphabet $\L$, the question posed by 
{\CPDlong} ({\CPD} in short) is whether the grid can be completely filled or not?

Two cases will be considered: whether each word is used at most once, or if each word can be assigned multiple times. 
In this article, we will sometimes suppose that some cells are pre-filled with some elements of $\L$. 
In this case, a solution is feasible if it is consistent with the pre-filled cells. 
Below we propose a first result when 
all the shared cells are pre-filled.
\begin{proposition} \label{prefilled:shared:cells}
{\CPD} and {\CPO} can be solved in polynomial time if all the shared cells in the grid are pre-filled, whether word reuse is allowed or not.
\end{proposition}

\begin{proof}
If word reuse is allowed,
then for each combination of letters placed in these cells, we greedily fill
out the rest of each slot with the maximum value word that can still be placed
there.  This is guaranteed to produce
the optimal solution. On the other hand, if word reuse is not allowed, we
construct a bipartite graph, with elements of $\D$ on one side and the slots on
the other, and place an edge between a word and a slot if the word can still be
placed in the slot. If we give each edge weight equal to the value of its
incident word reduced by the weight of the letters imposed by the shared cells
of the slot, then  an optimal solution corresponds to a maximum weight
matching. 
\end{proof}

One can associate a bipartite graph, hereafter called the {\em grid graph}, with each grid: each slot is a vertex and two vertices share an edge if the corresponding slots overlap. The grid (and then, the grid graph) is not necessarily connected.   

Let us also note that so far we have been a bit vague about the encoding of the problem. Concretely, we could use a simple representation which lists for each slot the coordinates of its first cell, its size, and whether the slot is horizontal or vertical; and then supplies a list of all words in the dictionary and an encoding of the weight function. Such a representation would allow us to perform all the basic operations needed by our algorithms in polynomial time, such as deciding if it is possible to place a word $\w$ in a slot $S$, and which letter would then be placed in any particular cell of $S$. However, one drawback of this encoding is that its size may not be polynomially bounded in $n+m$, as some words may be exponentially long. We can work around this difficulty by using a more succinct representation: we are given the same information as above regarding the $n$ slots; for each word we are given its total weight; and for each slot $S$ and word $\w$, we are told whether $\w$ fits exactly in $S$, and if yes, which letters are placed in the cells of $S$ which are shared with other slots. Since the number of shared cells is $O(n^2)$ this representation is polynomial in $n+m$ and it is not hard to see that we are still able to perform any reasonable basic operation in polynomial time and that we can transform an instance given in the simple representation to this more succinct form. Hence, in the remainder, we will always assume that the size of the input is polynomially bounded in $n+m$.




We will rely on the Exponential Time Hypothesis (ETH) of Impagliazzo, Paturi,
and Zane \cite{ImpagliazzoPZ01}, which states the following:

\begin{conjecture}\label{conj:eth} Exponential Time Hypothesis: there exists an
$\epsilon>0$, such that \textsc{3-SAT} on instances with $n$ variables and $m$
clauses cannot be solved in time $2^{\epsilon(n+m)}$.  \end{conjecture}

Note that it is common to use the slightly weaker formulation which states the
ETH as the assumption that \textsc{3-SAT} cannot be solved in time
$2^{o(n+m)}$.  This is known to imply that $k$-\textsc{Independent Set} cannot
be solved in time $n^{o(k)}$\cite{CyganFKLMPPS15}. We use this fact in
\cref{thm:byrows}. In \cref{sec:slots} we will rely on the randomized version
of the ETH, which has the same statement as \cref{conj:eth} but for randomized
algorithms with expected running time $2^{\epsilon(n+m)}$.

\section{When the Grid Graph is Tree-like} \label{sec:tree:like}

In this section we are considering instances of {\CPD} and {\CPO} where the grid graph is 
similar to a tree. First, we give an algorithm for both problems in cases where the grid graph 
has bounded treewidth and we are allowed to reuse words and we show that this algorithm is 
essentially optimal. Then, we show that {\CPD} and {\CPO} are much harder to deal with, in the
case where we are not allowed to reuse words, by proving that the problems are NP-hard 
even for instances where the grid graph is just a matching. 
For the instances such that {\CPD} is NP-hard, we know that {\CPO} is NP-hard. That happens because we can assume that all the letters have weight equal to $1$ hence a solution for {\CPD} is an optimal solution for {\CPO}.

\subsection{Word Reuse}



We propose a dynamic programming algorithm for {\CPO} and hence also for {\CPD}. Note that it can be extended to the case where some cells of the instance are pre-filled.

\begin{theorem} \label{DP:treewidth:CPO}
If we allow word reuse, then {\CPO} can be solved in time $(m+1)^{\tw}(n+m)^{O(1)}$ on inputs where $\tw$ is the treewidth of the grid graph.
\end{theorem}

\begin{proof}
As the techniques we are going to use are standard we are sketching some details. For more details on tree decomposition (definition and terminology) see \cite[Chap. 7]{CyganFKLMPPS15}. Assuming that we have a rooted nice tree decomposition of the grid graph, we are going to perform dynamic programming on the nodes of this tree decomposition. For a node $B_t$ of the given tree decomposition of the grid graph we denote by $B_t^{\downarrow}$ the set of vertices of the grid graph that appears in the nodes of the subtree with $B_t$ as a root. Since each vertex of the grid graph corresponds to a slot, we interchangeably mention a vertex of the grid graph and its corresponding slot. In particular, we say that a solution $\sigma$ assigns words to the vertices of the grid graph, and $\sigma(v)$ denotes the word assigned to $v$.

For each node $B_t$ of the tree decomposition we are going to keep all the triplets $(\sigma,W,W_t)$ such that:
\begin{itemize}
    \item $\sigma$ is an assignment of words to the vertices of $B_t$; 
    \item $W$ is the weight of $\sigma$ restricted to the vertices appearing in $B_t$; 
    \item and $W_{m}$ is the maximum weight, restricted to the vertices appearing in $B_t^{\downarrow}$, of an assignment consistent with $\sigma$. 
\end{itemize}
In order to create all the possible triplets for all the nodes of the tree decomposition we are going to explore the nodes 
from leaves to the root. Therefore, each time we visit a node we assume that we have already created the triplets for all its children. Let us explain how we deal with the different types of nodes. 

In the Leaf nodes we have no vertices so we keep an empty assignment ($\sigma$ does not assign any word) and the weights $W$ and $W_m$ are equal to $0$.

For an Introduce node $B_t$ we need to take in consideration its child node. Assume that $u$ is the introduced vertex; for each triplet $(\sigma, W, W_m)$ of the child node we are going to create all the triplets $(\sigma', W', W'_m)$ for the new node as follows. First we find 
all the words $\w \in \D$ that fit in the corresponding slot of $u$ and respect the assignment $\sigma$ (i.e., if there are cells that are already filled under $\sigma$ and $\w$ uses these cells then it must have the same letters). 
We create one triplet $(\sigma', W',W'_m)$ for each such a $\w$ as follows: 
\begin{itemize}
    \item We set $\sigma' (u) := \w$ and $\sigma' (v) := \sigma (v)$ for all $v \in B_t \setminus \{u\}$.
    \item We can easily calculate the total weight, $W'$, of the words in $B_t$ where the shared letters are counted only once under the assignment $\sigma'$. 
    \item For the maximum weight $W'_m$ we know that it is increased by the same amount as $W$; so we set $W'_m = W_m + W' - W$.
\end{itemize}
Observe that we do not need to consider the intersection with
slots whose 
vertices appear in $B_t^{\downarrow}\setminus B_t$ as each node of a tree decomposition is a cut set.

Finally, we need to take in consideration that we can leave a slot empty. For this case we create a new word $\w_*$ which, we assume that, fits in all slots and $\w_*$ has weight $0$.
Because the empty word has weight $0$, $W'$ and $W'_{m}$ are identical to $W$ and $W_{m}$ so for each triplet of the child node, we only need to extend 
$\sigma$ by assigning $\w_*$ to $u$. In the case we assign the empty word somewhere we will consider that the cells of this slot are empty unless another word $\w \neq \w_*$ uses them.

For the Forget nodes we need to restrict the assignments of the child node to the vertex set of the Forget node, as it has been reduced by one vertex (the forgotten vertex), and reduce the weight $W$ (which we can calculate easily). The maximum weight is not changed by the deletion. 

However, if we restrict the assignments we may end up with several triplets $(\sigma,W,W_m)$ with identical assignments $\sigma$. In that case we are keeping only the triplet with maximum $W_m$. Observe that we are allowed to keep only triplets with the maximum $W_m$ because each node of a tree decomposition is a cut set so the same holds for the Forget nodes. Specifically, the vertices that appear in the nodes higher than a Forget node $B_t$ of the tree decomposition do not have edges incident to vertices in $B_t^{\downarrow} \setminus B_t$ so we only care for the assignment in $B_t$. 

Finally, we need to consider the Join nodes. Each Join node has exactly two children. For each possible assignment $\sigma$ on the vertices of this Join node, 
we create a triplet iff this $\sigma$ appears in a triplet of both children of the Join node. 

Because 
$W$ is related only to the assignment $\sigma$, it is easy to see that it will be the same as in the children of the Join node. So we need to find the maximum weight $W_m$. Observe that between the vertices that appear in the subtrees of two children of a Join node there are no edges except those incident to the vertices of the Join node. Therefore, we can calculate the maximum weight $W_m$ as follows: first we consider the maximum weight of each child of the Join node reduced by $W$, we add all these weights and, in the end, we add again the $W$. It is easy to see that this way we consider the weight of the cells appearing in each subtree without those of the slots of the Join node and we add the weight of the words assigned to the vertices of the Join node in the end. 

For the running time we need to observe that the number of nodes of a nice tree decomposition is $O( \tw \cdot n)$ and all the other calculations are polynomial in $n+m$ so we only need to consider the different assignments for each node. Because for each vertex we have $|\D|+1$ choices, 
the number of 
different assignments for a node is at most $(|\D|+1)^{\tw+1}$. 
\end{proof}



It seems that the algorithm we propose for {\CPD} is essentially optimal,
even if we consider a much more restricted case.

\begin{theorem}\label{thm:byrows} {\CPD} 
with word reuse is W[1]-hard parameterized
by the number of horizontal slots of the grid, even for alphabets with two letters.  Furthermore, under the ETH, no algorithm can solve this
problem in time $m^{o(k)}$, where $k$ is the number of horizontal slots. \end{theorem}
\begin{proof}

We perform a reduction from $k$-\textsc{Independent Set}, where we are given a
graph $G=(V,E)$ with $|V|$ vertices and $|E|$ edges and are looking for an
independent set of size $k$.  This problem is well-known to be W[1]-hard and
not solvable in $|V|^{o(k)}$ time under the ETH \cite{CyganFKLMPPS15}. We assume
without loss of generality that $|E|\neq k$.  Furthermore, we can safely assume
that $G$ has no isolated vertices.

We first describe the grid of our construction which fits within an area of
$2k-1$ lines and $2|E|-1$ columns.  We construct:

\begin{enumerate}

\item $k$ horizontal slots, each of length $2|E|-1$ (so each of these slots is as
long horizontally as the whole grid).  We place these slots in the unique way
so that no two of these slots are in consecutive lines. We number these
horizontal slots $1,\ldots,k$ from top to bottom.

\item $|E|$ vertical slots, each of length $2k-1$ (so each of these slots is long
enough to cover the grid top to bottom). We place these slots in the unique way
so that no two of them are in consecutive columns. We number them $1,\ldots,|E|$
from left to right. 

\end{enumerate}

Before we describe the dictionary, let us give some intuition about the grid.
The main idea is that in the $k$ horizontal slots we will place $k$ words that
signify which vertices we selected from the original graph. Each vertical slot
represents an edge of $E$, and we will be able to place a word in it if and
only if we have not placed words representing two of its endpoints in the
horizontal slots. 

Our alphabet has two letters, say $0,1$. In the remainder, we assume that the edges of the
original graph are numbered, that is, $E=\{e_1,\ldots,e_{|E|}\}$. The dictionary is
as follows:

\begin{enumerate}

\item For each vertex $v$ we construct a word of length $2|E|-1$. For each
$i\in\{1,\ldots,|E|\}$, if the edge $e_i$ is incident on $v$, then the letter at
position $2i-1$ of the word representing $v$ is $1$. All other letters of the
word representing $v$ are $0$. Observe that this means that if $e_i$ is
incident on $v$ and we place the word representing $v$ on a horizontal slot,
the letter $i$ will appear on the $i$-th vertical slot. Furthermore, the word
representing $v$ has a number of $1$s equal to the degree of $v$.

\item We construct $k+1$ words of length $2k-1$. One of them is simply
$0^{2k-1}$. The remaining are $0^{2j-2}10^{2k-2j}$, for $j\in\{1,\ldots,k\}$,
that is, the words formed by placing a $1$ in an odd-numbered position and $0$s
everywhere else. Observe that if we place one of these $k$ words on a vertical
slot, a $1$ will be placed on exactly one horizontal slot.

\end{enumerate}

This completes the construction. We now observe that the $k$ horizontal slots
correspond to a vertex cover of the grid-graph. Therefore, if the reduction
preserves the answer, the hardness results for $k$-\textsc{Independent Set}
transfer to our problem, since we preserve the value of the parameter.

We claim that if there exists an independent set of size $k$ in $G$, then it is
possible to fill the grid. Indeed, take such a set $S$ and for each $v\in S$ we
place the word representing $v$ in a horizontal slot. Consider the $i$-th
vertical slot. We will place in this slot one of the $k+1$ words of length
$2k-1$. We claim that the vertical slot at this moment contains the letter $1$
at most once, and if $1$ appears it must be at an odd position (since these are
the positions shared with the horizontal slots). If this is true, clearly there
is a word we can place. To see that the claim is true, recall that since $S$ is
an independent set of $k$ distinct vertices, there exists at most one vertex in
$S$ incident on $e_i$.

For the converse direction, recall that $|E|\neq k$.  This implies that if there
is a way to fill out the whole grid, then words representing vertices must go
into horizontal slots and words of length $2k-1$ must go into vertical slots.
By looking at the words that have been placed in the horizontal slots we obtain
a collection of $k$ (not necessarily distinct) vertices of $G$. We will prove
that these vertices must actually be an independent set of size exactly $k$. To
see this, consider the $i$-th vertical slot. If our collection of vertices
contained two vertices incident on $e_i$, it would have been impossible to fill
out the $i$-th vertical slot, since we would need a word with two $1$s.
Observe that the same argument rules out the possibility that our collection
contains the same vertex $v$ twice, as the column corresponding to any edge
$e_i$ incident on $v$ would have been impossible to fill.  \end{proof}

\subsection{No Word Reuse}

If a word cannot be reused, then {\CPD} looks more challenging. Indeed, in the following theorem we prove
that if reusing words is not allowed, then the problem becomes NP-hard even if
the grid graph is acyclic and the alphabet size is $2$. (Note that if the
alphabet size is $1$, the problem is trivial, independent of the structure of
the graph).

\begin{theorem} \label{thm:forest:2letters} {\CPD} 
is NP-hard, even for instances
where all of the following restrictions apply: (i) the grid graph is a union of stars
(ii) the alphabet contains only two letters (iii) words cannot be reused. 
\end{theorem}


\begin{proof}

We show a reduction from \textsc{3-Partition}. Recall that in
\textsc{3-Partition} we are given a collection of $3n$ distinct positive
integers $x_1,\ldots,x_{3n}$ and are asked if it is possible to partition these
integers into $n$ sets of three integers (triples), such that all triples have
the same sum. This problem has long been known to be strongly NP-hard
\cite{GJ79} and NP-hardness when the integers are distinct was shown by Hulett
et al. \cite{HulettWW08}.  We can assume that $\sum_{i=1}^{3n} x_i = nB$ and
that if a partition exists each triple has sum $B$. Furthermore, we can assume
without loss of generality that $x_i>6n$ for all $i\in \{1,\ldots,3n\}$
(otherwise, we can simply add $6n$ to all numbers and adjust $B$ accordingly
without changing the answer).

Given an instance of \textsc{3-Partition} as above, we construct a crossword
instance as follows. First, the alphabet only contains two letters, say the
letters $*$ and $!$. To construct our dictionary we
do the following:


\begin{enumerate}

\item For each $i\in\{1,\ldots,3n\}$, we add to the dictionary one word of
length $x_i$ that begins with $!$ and $n-1$ words of length $x_i$ that begin
with $*$. The remaining letters of these words are chosen in an arbitrary way
so that all words remain distinct.

\item For each $i,j,k\in\{1,\ldots,3n\}$ with $i<j<k$ we check if
$x_i+x_j+x_k=B$. If this is the case, we add to the dictionary the word
$*^{2i-2}!*^{2j-2i-1}!*^{2k-2j-1}!*^{6n-2k}$. In other words, we constructed a
word that has $*$ everywhere except in positions $2i-1, 2j-1$, and $2k-1$. The
length of this word is $6n-1$. Let $f$ be the number of words added to the
dictionary in this step. We have $f \le \genfrac(){0pt}{2}{3n}{3} = O(n^3)$.
\end{enumerate}

We now also need to specify our grid. We first construct $f$ horizontal slots,
each of length $6n-1$. Among these $f$ slots, we select $n$, which we call the
``interesting'' horizontal slots. For each interesting horizontal slot, we
construct $3n$ vertical slots, such that the $i$-th of these slots has length
$x_i$ and its first cell is the cell in position $2i-1$ of the interesting
horizontal slot. This completes the construction, which can clearly be carried
out in polynomial time. Observe that the first two promised restrictions are
satisfied as we have an alphabet with two letters and each
vertical slot intersects at most one horizontal slot (so the grid graph is a
union of stars).


We claim that if there exists a partition of the original instance, then we can
place all the words of the dictionary on the grid. Indeed, for each
$i,j,k\in\{1,\ldots,3n\}$ such that $\{x_i,x_j,x_k\}$ is one of the triples of
the partition, we have constructed a word of length $6n-1$ corresponding to the
triple $(i,j,k)$, because $x_i+x_j+x_k=B$. We place each of these $n$ words on
an interesting horizontal slot and we place the remaining words of length
$6n-1$ on the non-interesting horizontal slots. Now, for every
$i\in\{1,\ldots,3n\}$ we have constructed $n$ words, one starting with $!$ and
$n-1$ starting with $*$. We observe that among the interesting horizontal
slots, there is one that contains the letter $!$ at position $2i-1$ (the one
corresponding to the triple containing $x_i$ in the partition) and $n-1$
containing the letter $*$ at position $2i-1$. By construction, the vertical
slots that begin in these positions have length $x_i$. Therefore, we can place
all $n$ words corresponding to $x_i$ on these vertical slots. Proceeding in
this way we fill the whole grid, fulfilling the third condition.

For the converse direction, suppose that there is a way to fill the whole grid.
Then, vertical slots must contain words that were constructed in the second
step and represent integers $x_i$, while horizontal slots must contain words
constructed in the first step (this is a consequence of the fact that $x_i>6n$
for all $i\in\{1,\ldots,3n\}$). We consider the $n$ interesting horizontal
slots. Each such slot contains a word that represents a triple $(i,j,k)$ with
$x_i+x_j+x_k=B$. We therefore collect these $n$ triples and attempt to
construct a partition from them. To do this, we must prove that each $x_i$ must
belong to exactly one of these triples. However, recall that we have exactly
$n$ words of length $x_i$ (since all integers of our instance are distinct) and
exactly $n$ vertical slots of this length. We conclude that exactly one 
 vertical slot must have $!$ as its first letter, therefore $x_i$ appears
in exactly one triple and we have a proper partition.  \end{proof}

Actually, the problem remains NP-hard even in the case where the grid graph is a matching and 
the alphabet contains three letters. 
This is proved for grid graphs composed of $\T$s, where a $\T$ is a horizontal slot solely intersected by the first cell of a vertical slot. 

\begin{theorem} \label{thm:matching:3letters}  {\CPD} is NP-hard, even for instances
where all of the following restrictions apply: (i) each word can be used only once (ii) the grid is consisted only by $\T$s and
(iii) the alphabet contains only three letters.
\end{theorem}


In order to prove this theorem we need first to define a restricted version of \textsc{Exactly-1  3-SAT}.

\begin{definition}[\textsc{Restricted Exactly 1 (3,2)-SAT}]
Assume that $\phi$ is a CNF formula where each clause has either three or two literals and each variable appears at most three times. We want to determine whether there exists a satisfying assignment so that each clause has exactly one true literal.
\end{definition}
\begin{lemma}
The \textsc{Restricted Exactly-1 (3,2)-SAT} is NP-complete.
\end{lemma}
\begin{proof}
We show a reduction from \textsc{Exactly-1  3-SAT} which is known to be NP-complete \cite{GJ79} (\textsc{lo4}, \textsc{one-in-three 3sat}).


Let $I = (\phi,X)$ be an instance of \textsc{Exactly-1  3-SAT} with $|X|=n$ variables and $m$ clauses. If there exists a variable $x$ with $k>3$ appearances, we replace each appearance with a fresh variable $x_i$, $i\in [k]$ and add to the formula the clauses $(\neg
x_1 \lor x_2)\land (\neg x_2\lor x_3)\ldots (\neg x_{k}\lor x_1)$. We repeat
this for all variables that appear more than three times. Let $I'=(\phi',X')$ be this new instance. 

We claim that $I = (\phi,X)$ is a yes instance of \textsc{Exactly-1  3-SAT} iff $I' = (\phi',X')$ is a yes instance of \textsc{Restricted Exactly-1  (3,2)-SAT}. 

Let $S:X\rightarrow \{T,F\}$ be a satisfying assignment for $\phi$ such that each clause of $\phi$ has exactly one true literal. It is not hard to see that 
$S':X'\rightarrow \{T,F\}$ such that $S'(x)=S(x)$ if $x \in X$ and $S'(x_i)=S(x)$ if $x_i$ replaces one appearance of $x \in X$, 
is a satisfying assignment for $\phi'$ such that each clause of $\phi'$ has exactly one true literal.

Conversely, let $S':X'\rightarrow \{T,F\}$ be a satisfying assignment for $\phi'$ such that each clause of $\phi'$ has exactly one true literal. Let $x_i$, $i \in [k]$, be the variables 
replacing 
$x$. Because we have clauses $(\neg
x_1 \lor x_2)\land (\neg x_2\lor x_3)\ldots (\neg x_{k}\lor x_1)$ we know that all the $x_i$, $i \in [k]$, must have the same value in order to guarantee 
that all of these clauses 
have exactly one true literal. Furthermore, is not hard to see that 
$S :X\rightarrow \{T,F\}$ where $S(x)=S'(x)$ if $x \in X'$ and $S(x)=S'(x_1)$ if $x_1$ 
replaces one appearance of $x$, then 
$S$ is a satisfying assignment for $\phi$ such that each clause of $\phi$ has exactly one true literal.
\end{proof}

Now, let us give a construction that we are going to use.

\noindent
\textbf{Construction.}\\
\noindent
Let $\phi$ be an instance of \textsc{Restricted Exactly 1 (3,2)-SAT} with variables $X=\{x_1,\ldots,x_n\}$ and clauses 
$C=\{c_1,\ldots,c_m\}$. We will construct an instance of the crossword problem with 
alphabet $\L=\{s_1,s_2,s_3\}$ 
where each letter has weight 1. 
The dictionary $\D$ is as follows.

Let $nl_j \in \{2,3\}$ be the number of literals in  $c_j$. For each variable $x_i$, let $a_i \le 3$ be the number of its appearances in $\phi$.  Then, we create $3a_i$ words, $\w_{i,k,T}$, $\w_{i,k,F}$ and $\w_{i,k}$, for each $k \in [a_i]$ as follows. 
\begin{itemize}
    \item $\w_{i,k,T}$ and $\w_{i,k,F}$ have length $m+n+3i+k$,
    \item the last letter of $\w_{i,k,T}$ is $s_{k}$,
    \item the last letter of $\w_{i,k,F}$ is $s_{k'}$ where $k':=k+1$ when $k<a_i$, otherwise $k':=1$, 
    \item if the $k$-th appearance of $x_i$ is positive then, $\w_{i,k,T}$ starts with $s_1$ and $\w_{i,k,F}$ starts with $s_2$,
    \item if the $k$-th appearance of $x_i$ is negative then, $\w_{i,k,T}$ starts with $s_2$ and $\w_{i,k,F}$ starts with $s_1$, 
    \item the word $\w_{i,k}$ has length $m+i+1$ and  starts with $s_k$, and 
    \item all the other letters of these words can be chosen arbitrarily.
\end{itemize}
Observe that the above process gives 
three words for each literal in $\phi$.

For each clause $c_j$, $j \in [m]$, we construct $nl_j$ distinct words $\w_j^t$, $t \in [nl_j]$ 
of length $1+j$ such that one of them starts with the letter $s_2$, the other $nl_j-1$ words start with $s_1$, and the unspecified letters can be chosen arbitrarily. 
Observe that we have enough positions in order to create $nl_j-1$ distinct words starting with $s_1$, which indicates that we can create $nl_j$ pairwise distinct words for each $c_j$.

In order to finish our construction we have to specify the grid. For each clause $c_j$ and each literal $l$ in $c_j$ we construct two pairs of slots as follows.
Let $l$ be the $k$-th appearance of variable $x_i$, $k\in [a_i]$. The first pair of slots  (type 1) consists of one horizontal slot $hSlot_{j,1}^{i,k}$  of length $m+n+3i+k$, and one vertical slot $vSlot_{j,1}^{i,k}$ of length $m+i+1$ such that, the last cell of the horizontal slot and the first cell of the vertical slot is the shared cell. 
The second pair of slots (type 2) consists of one horizontal slot $hSlot_{j,2}^{i,k}$ of length $m+n+3i+k$, and one vertical slot $vSlot_{j,2}^{i,k}$ of length $j+1$, that share their first cells. Here let us mention that the grid we constructed is consisted only by $\T$s.

Before we continue with the proof let us observe that in 
the instance of crossword puzzle we created the number of slots in the grid is equal to the number of words in the dictionary.
Furthermore, we can specify in which slots each word can be assigned by considering the size of the words and slots. 
For any $i \in [n]$ and $k \in [a_i]$ the word $\w_{i,k}$ can be assigned only to the vertical slots of the type $1$ pairs of slots. For any $j \in [m]$ and $t \in [nl_j]$ the word $\w_j^t$ can be assigned 
only to the vertical slots of the type $2$ pairs of slots. 
The rest of the words can be assigned to horizontal slots of any type.


Let us first prove the following property 
where $j(i,k)$ denotes the index of the clause where the $k$-th occurrence of $x_i$ appears. 

\begin{property} \label{vetical_slots_x_i}
For any given $i\in [n]$, 
slots $hSlot_{j(i,k),1}^{i,k}$ and $vSlot_{j(i,k),1}^{i,k}$ for $k\in [a_i]$ are all filled iff we have assigned either all the words of $\{ \w_{i,k,T}: k\in [a_i]\}$, or all the words of $\{ \w_{i,k,F} : k\in [a_i]\}$, to the slots $hSlot_{j(i,k),1}^{i,k}$, 
$k\in [a_i]$.
\end{property}

\begin{proof}
In one direction, if we have assigned to 
slots $hSlot_{j(i,k),1}^{i,k}$, $k \in [a_i]$, all the words of $\{ \w_{i,k,T} : k \in [a_i]\}$ or all the words of $\{ \w_{i,k,F} : k \in [a_i]\}$, then  
all the letters $s_{1}, \ldots, s_{a_i}$ 
appear exactly once in the end of these $a_i$ slots. Because the words of $\{ \w_{i,k}$: $k \in [a_i]\}$ start exactly with this set of letters, 
there is a unique way to assign them properly to the slots $vSlot_{j(i,k),1}^{i,k}$, $k \in [a_i]$. 

Conversely, assume that all the type 1 pairs of slots of $x_i$ are filled. Because the only words that have the same length as 
slots $vSlot_{j(i,k),1}^{i,k}$, $k \in [a_i]$,  
are the words of $\{ \w_{i,k} : k \in [a_i]\}$, we know that in the end of 
slots $hSlot_{j(i,k),1}^{i,k}$, $k \in [a_i]$, each letter of $\{s_{1}, \ldots, s_{a_i}\}$ appears exactly once.  It is not hard to see that no combination of words except $\{ \w_{i,k,T}:k \in [a_i]\}$ or $\{ \w_{i,k,F} : k \in [a_i]\}$, gives the same letters in the shared positions. \end{proof}

Now we are ready to present the proof of Theorem~\ref{thm:matching:3letters}.

\begin{proof}

We show a reduction from \textsc{Restricted Exactly 1  (3,2)-SAT}. 
We claim that $\phi$ is a yes instance of  \textsc{Restricted Exactly 1 (3,2)-SAT} iff we can fill all the slots of the grid.

Suppose $f : X\rightarrow \{T,F\}$ is a truth assignment so that each clause of $\phi$ has exactly one true literal that satisfies $\phi$. 

We are going to show a way to fill all the slots of the grid. 
Each variable $x_i$ appears in $a_i$ literals;
let $l(i,k)$, $k \in [a_i]$, be these literals and $j(i,k) \in [m]$, $k \in [a_i]$, be the indices of the clauses $c_{j(i,k)}$ that contain the corresponding literals.

For each variable $x_i$, 
fill the $3 a_i$ slots
$hSlot_{j(i,k),1}^{i,k}$, $hSlot_{j(i,k),2}^{i,k}$ and $vSlot_{j(i,k),1}^{i,k}$ for all $k\in [a_i]$ as follows. 
If $f(x_i)=T$, then:
\begin{itemize}
    \item assign $\w_{i,k,T}$ to $hSlot_{j(i,k),1}^{i,k}$ for all $k \in [a_i]$ and
    \item assign $\w_{i,k,F}$ to $hSlot_{j(i,k),2}^{i,k}$ for all $k \in [a_i]$.
\end{itemize}
\noindent
Otherwise ($f(x_i)=F$):
\begin{itemize}
    \item assign $\w_{i,k,F}$ to $hSlot_{j(i,k),1}^{i,k}$ for all $k \in [a_i]$ and
    \item assign $\w_{i,k,T}$ to $hSlot_{j(i,k),2}^{i,k}$ for all $k \in [a_i]$.
\end{itemize}
Finally, in both cases, we assign the words of $\{ \w_{i,k} : k \in [a_i]\}$ to the slots $vSlot_{j(i,k),1}^{i,k}$ for $k \in [a_i]$ in any way they fit.

In order to fill the grid completely, for each $j\in [m]$, we assign to the $nl_{j}$ slots, $vSlot_{j,2}^{i,k}$, the words $\w_{j}^{k'}$ for $k' \in [nl_j]$ in any way they fit.

It is not hard to see that we have assigned words to slots of the same length. 
It remains to prove that the words we have assigned have the same letters in the shared positions. 

First observe that for a variable $x_i$ and the slots $hSlot_{j(i,k),1}^{i,k}$, $k \in [a_i]$, we have put either 
$\{ \w_{i,k,T} : k \in [a_i]\}$ or 
$\{ \w_{i,k,F} : k \in [a_i]\}$. Therefore, we know by Property~\ref{vetical_slots_x_i} that we can use the words of $\{ \w_{i,k} : k \in [a_i]\}$ in the slots $vSlot_{j(i,k),1}^{i,k}$, $k \in [a_i]$.


In the $nl_{j}$ slots, $vSlot_{j,2}^{i,k}$, related to clause $c_j$, we have put the words $\w_{j}^{k'}$, $k'\in [nl_j]$. One of these words starts with $s_2$ and the $nl_j-1$ others start with $s_1$. We will show that the same holds for the words we have assigned in the $nl_{j}$ slots $hSlot_{j,2}^{i,k}$.


Observe that each literal $l \in c_j$ can be described by a unique triplet $(j,i,k)$ where $j\in [m]$ is the index of the clause, 
$i\in [n]$ is the index of the variable $x_i$ on which $l$ is built, and $k\in [a_i]$ is the number of times that $x_i$ has appeared in $\phi$ until now. We claim that if the literal $l$ described by $(j,i,k)$ satisfies $c_j$, then the word 
assigned to $hSlot_{j,2}^{i,k}$ starts with $s_2$, otherwise it starts with $s_1$. 

If 
$l$ satisfies 
$c_j$, then either $l=x_i$ and $f(x_i)=T$ or $l=\neg x_i$ and $f(x_i)=F$. If $l=x_i$ (resp., $l= \neg x_i$), then we have assigned 
$\w_{i,k,F}$ (resp., $\w_{i,k,T}$) to $hSlot_{j,2}^{i,k}$ which starts with $s_2$ because $f(x_i)=T$ (resp., $f(x_i)=F$). If $l$ does not satisfy 
$c_j$, then we used 
$\w_{i,k,T}$ (resp., $\w_{i,k,F}$) which starts with $s_1$. 

Finally, because we assumed that each clause is satisfied by exactly one literal, we know that one of the clause words starts with $s_2$ and the other $nl_j-1$ clause words start with $s_1$.

Conversely, we claim that if we can fill the whole grid, then we can construct a truth assignment $f:X\rightarrow \{T,F\}$ such that each clause of $\phi$ has exactly one true literal. 
Furthermore, one such 
assignment is the following:
\begin{align} \label{assignmentproof3letters}
    f(x_i)=\begin{cases} T, \text{ if } \w_{i,1,T} \text{ is assigned to } hSlot_{j(i,1),1}^{i,1},\\
    F, \text{otherwise}.
    \end{cases}
\end{align}

We first prove the following claim.

\begin{claim}
Let 
$l$ be the literal of a clause $c_j$ corresponding to the $k$-th appearance of some variable $x_i$. 
$l$ is true under the truth 
assignment 
(\ref{assignmentproof3letters}) 
iff the word in 
$hSlots_{j,2}^{i,k}$ starts with 
$s_2$.
\end{claim}


\begin{claimproof} Due to its length, 
$hSlots_{j,2}^{i,k}$ 
receives 
either $\w_{i,k,T}$ or $\w_{i,k,F}$, and one of these words starts with $s_2$ whereas the other starts with $s_1$. Therefore, we have two cases. In the first case $\w_{i,k,F}$ starts with $s_2$, then $\w_{i,k,T}$ starts with $s_1$ and $l=x_i$.  In the second case, $\w_{i,k,T}$ starts with $s_2$, $\w_{i,k,F}$ starts with $s_1$ and $l=\neg x_i$.

Assume that $\w_{i,k,F}$ (resp., $\w_{i,k,T}$) starts with $s_2$. By construction, we have that $l=x_i$ (resp., $l=\neg x_i$).

If  $\w_{i,k,F}$ (resp., $\w_{i,k,T}$) is assigned to 
$hSlots_{j,2}^{i,k}$, then  
$\w_{i,k,T}$ (resp., $\w_{i,k,F}$) is assigned to 
$hSlots_{j,1}^{i,k}$. By Property~\ref{vetical_slots_x_i} we know that 
$hSlots_{j,1}^{i,1}$ must contain 
$\w_{i,1,T}$ (resp., $\w_{i,1,F}$) so $f(x_i)= T$ (resp., $f(x_i)= F$). So, if  
$\w_{i,k,F}$ (resp., $\w_{i,k,T}$) is assigned to 
$hSlots_{j,2}^{i,k}$, then we know that $f(x_i)= T$ (resp., $f(x_i)= F$) and  $l=x_i$ (resp., $l=\neg x_i$) which means that $l$ must be true under the truth 
assignment
(\ref{assignmentproof3letters}).

In reverse direction, if we have assigned $\w_{i,k,T}$ (resp., $\w_{i,k,F}$) to 
$hSlots_{j,2}^{i,k}$, then we know that $f(x_i)= F$ (resp., $f(x_i)= T$) and  $l=x_i$ (resp., $l=\neg x_i$) thus, $l$ is false under the truth assignment (\ref{assignmentproof3letters}). \end{claimproof}

Based on the previous claim, we will show that each clause has exactly one true literal under the truth assignment $f$ given in 
(\ref{assignmentproof3letters}).

For any $j \in [m]$ there are exactly $nl_j$ pairs $(i,k)$ where $i \in [n]$ and $k \in [a_i]$ such that the $k$-th appearance of $x_i$ is in 
$c_j$. Let $C_j$ be the set that contains contains all these pairs $(i,k)$.

Observe that for each pair $(i,k)\in C_j$ there exists a pair of slots  $hSlots_{j,2}^{i,k}$, $vSlots_{j,2}^{i,k}$ which 
share their first cells. 
Because the grid is full, 
the $nl_j$ vertical slots, $vSlots_{j,2}^{i,k}$, where $(i,k) \in C_j$, must contain the words $\w_{j}^t$, $t \in [nl_j]$. One of these words starts with $s_2$ and $nl_j-1$ others start with $s_1$.
Therefore, the same must hold for the words that have been assigned in the slots $hSlots_{j,2}^{i,k}$ for $(i,k) \in C_j$.

Using the previous claim, we know that 
one of the literals in $c_j$ is true and the other $nl_j-1$ are false under the truth assignment \ref{assignmentproof3letters}. Therefore, if we can fill the whole grid, then there exists a truth assignment such that exactly one literal of each clause of $\phi$ is true. \end{proof}
\begin{remark}
In our construction each $\T$ has unique shape\footnote{Two crosses are of the same shape if they are identical: same number of horizontal cells, same number of vertical cells, and same shared cell.} so the problem remains $NP$-hard even in this case.  
\end{remark}



\begin{remark} \label{remark:byrows:without:reuse}
Theorem~\ref{thm:byrows} can be adjusted to work also for the
case where word reuse is not allowed. We simply need to add a suffix of length $\log m$
to all words of length $2k-1$ and add rows to the grid accordingly. Hence, under the ETH, no algorithm can solve this problem in time $m^{o(k)}$, where $k$ is the number of horizontal slots.
\end{remark}


Finally, observe that by filling the slots of a vertex cover of the grid graph, 
all the shared cells are pre-filled. Since there are at most $m^k$ (where $k$ is the size of the vertex cover)
ways to assign words to these slots, by Proposition~\ref{prefilled:shared:cells}, we get the following corollary.



\begin{corollary}
Given a vertex cover of size $k$ of the grid graph we can solve {\CPD} and {\CPO} in time $m^{k}(n+m)^{O(1)}$. Furthermore, as vertex cover we can take the set of horizontal slots.
\end{corollary}
Therefore, the bound given in \cref{remark:byrows:without:reuse} for the parameter vertex cover is tight.

\section{Parameterized by Total Number of Slots}\label{sec:slots}

In this section we consider a much more restrictive parameterization of the
problem: we consider instances where the parameter is $n$, the total number of
slots. Recall that in \cref{thm:byrows} (and \cref{remark:byrows:without:reuse}) we already considered the complexity of
the problem parameterized by the number of \emph{horizontal} slots of the
instance. We showed that this case of the problem cannot be solved in
$m^{o(k)}$ and that an algorithm with running time roughly $m^k$ is possible
whether word reuse is allowed or not. 


Since parameterizing by the number of horizontal slots is not sufficient to render the problem FPT, we therefore consider our parameter to be the total number of slots. This is, finally, sufficient to obtain a simple FPT algorithm.

\begin{corollary}\label{corollary:algslots} There is an algorithm that solves {\CPD}
and {\CPO} in time $O^*(\ell^{n^2/4})$, where $n$ is the total number of slots
and $\ell$ the size of the alphabet, whether word reuse is allowed or not.
\end{corollary}

\begin{proof}
Since there are $n$ slots in the instance, even if the grid 
is a complete bipartite graph, the instance contains at most $n^2/4$ cells
which are shared between two slots. In time $\ell^{n^2/4}$ we consider all
possible letters that could be placed in these cells. Finally, as we have shown in Proposition~\ref{prefilled:shared:cells}, each of these instances can be solved in polynomial time.
\end{proof}

Even though the running time guaranteed by \cref{corollary:algslots} is FPT for
parameter $n$, we cannot
help but observe that the dependence on $n$ is rather disappointing, as our
algorithm is exponential \emph{in the square} of $n$. It is therefore a natural
question whether an FPT algorithm for this problem can achieve complexity
$2^{o(n^2)}$, assuming the alphabet size is bounded. The main result of this
section is to establish that this is likely to be impossible.

\subparagraph*{Overview} Our hardness proof consists of two steps. In the first
step we reduce \textsc{3-SAT} to a version of the same problem where variables
and clauses are partitioned into $O(\sqrt{n+m})$ groups, which we call
\textsc{Sparse 3-SAT}. The key property of this intermediate problem is that
interactions between groups of variables and groups of clauses are extremely
limited. In particular, for each group of variables $V_i$ and each group of
clauses $C_j$, at most one variable of $V_i$ appears in a clause of $C_j$. We
obtain this rather severe restriction via a randomized reduction that runs in
expected polynomial time. The second step is to reduce \textsc{Sparse 3-SAT} to
{\CPD}. Here, every horizontal slot will represent a group of variables and
every vertical slot a group of clauses, giving $O(\sqrt{n+m})$ slots in total.
Hence, an algorithm for {\CPD} whose dependence on the total number of slots is
subquadratic in the exponent will imply a sub-exponential time (randomized)
algorithm for \textsc{3-SAT}. The limited interactions between groups of
clauses and variables will be key in allowing us to execute this reduction
using a \emph{binary} alphabet.

Let us now define our intermediate problem.

\begin{definition}\label{def:sparse} In \textsc{Sparse 3-SAT} we are given an
integer $n$ which is a perfect square and a \textsc{3-SAT} formula $\phi$ with
at most $n$ variables and at most $n$ clauses, such that each variable appears
in at most $3$ clauses. Furthermore, we are given a partition of the set of
variables $V$ and the set of clauses $C$ into $\sqrt{n}$ sets
$V_1,\ldots,V_{\sqrt{n}}$ and $C_1,\ldots,C_{\sqrt{n}}$ of size at most
$\sqrt{n}$ each, such that for all $i,j\in [\sqrt{n}]$ the number of variables
of $V_i$ which appear in at least one clause of $C_j$ is at most one.
\end{definition}

Now, we are going to prove the hardness of \textsc{Sparse 3-SAT}, which is the
first step of our reduction.

\begin{lemma}\label{lem:sparse-hard}Suppose the randomized ETH is true. Then,
there exists an $\epsilon>0$ such that \textsc{Sparse 3-SAT} cannot be solved
in time $2^{\epsilon n}$. \end{lemma}

The first step of our reduction will be to prove that \textsc{Sparse 3-SAT}
cannot be solved in sub-exponential time (in $n$) under the randomized ETH, via
a reduction from \textsc{3-SAT}. To do this, we will need the following
combinatorial lemma.

\begin{lemma}\label{lem:color}

For each $\epsilon>0$ there exists $C>0$ such that for sufficiently large $n$ we have the following. There exists a randomized algorithm running in expected polynomial time which, given a bipartite graph $G=(A,B,E)$ such that
$|A|=|B|=n$ and the maximum degree of $G$ is $3$, produces a set $V'\subseteq A\cup B$ with $|V'|\ge 2(1-\epsilon)n$ and a coloring $c:V'\to [k]$ of the vertices of $V'$ with $k$ colors, where $k\le C\sqrt{n}$, such that for all $i\in [k]$ we have $|c^{-1}(i)|\le\sqrt{n}$ and
for all $i,j\in [k]$ the graph induced by $c^{-1}(i)\cup c^{-1}(j)$ contains at
most one edge.   \end{lemma}

\begin{proof}

Let $k = C \lceil \sqrt{n} \rceil$, where $C$ is a sufficiently large constant
(depending only on $\epsilon$) to be specified later. We color each vertex of
the graph uniformly at random from a color in $[k]$, call this coloring $c$.
Let $X_{i,j}$ be the set of edges which have as endpoints a vertex of color $i$
and a vertex of color $j$. 

Our algorithm is rather simple: initially, we set $V'=V$. Then, for each
$i,j\in [k]$ we check whether $X_{i,j}$ contains at most one edge. If yes, we
do nothing; if not, we select for each edge $e\in X_{i,j}$ an arbitrary
endpoint and remove that vertex from $V'$. In the end we return the set $V'$
that remains and its coloring. It is clear that this satisfies the property
that $c^{-1}(i)\cup c^{-1}(j)$ contains at most one edge for the graph induced
by $V'$ for all $i,j\in [k]$, so what we need to argue is that (i)
$|c^{-1}(i)|\le \sqrt{n}$ for all $i$ with high probability and (ii) that $V'$
has the promised size with at least constant probability.  If we achieve this
it will be sufficient to repeat the algorithm a polynomial number of times to
obtain the claimed properties with high probability, hence we will have an
expected running time polynomial in $n$.

For the first part, fix an $i\in [k]$ and observe that $E[ |c^{-1}(i)| ]
\le\frac{2\sqrt{n}}{C}$.  To prove that all $|c^{-1}(i)|$ are of size at most $4
\sqrt{n}/C$ with high probability (and hence also at most $\sqrt{n}$ for $C$
sufficiently large), we will use Chernoff's Inequality.

\begin{proposition}[Chernoff's Inequality]

Let $X$ be a binomial random variable and $\epsilon >0$. Then $P[|X-E[X]| > \epsilon E[X]] < 2e^{-\epsilon^2 E[X]/3}$
\end{proposition}

We take $\epsilon = 1$. It follows
that $P[|c^{-1}(i)| > 4 \sqrt{n}/C] \leq 2 e^{-2\sqrt{n}/3C}$. Now, taking the union bound, we
obtain that almost surely for all color $i$, $|c^{-1}(i)| < 4 \sqrt{n}/C$ 

The more interesting part of this proof is to bound the expected size of $V'$.
Let $e$ be an edge whose endpoints are colored with colors $i$ and $j$. We say
that $e$  is \emph{good} if no other edge in $G$ has one endpoint colored $i$
and the other colored $j$ by the coloring $c$. Let $u$ and $v$ be the endpoints
of $e$.  The probability of another edge having endpoints of colors $i$ and $j$
in the graph $G-\{u,v\}$ is at most $\frac{2|E|}{C^2 n} \leq \frac{6}{C^2}$.
The probability that at least one of the at most four edges incident to $e$ has
endpoints colored $i$ and $j$ is at most $\frac{4}{C\sqrt n}$.  Thus, the
probability that $e$ is good is at least $ 1 - \frac{6}{C^2} - \frac{4}{C\sqrt
n} > 1 - \frac{7}{C^2}$, if $n$ is sufficiently large.  Let $X$ be the number
of edges which are not good.  Then, $E[X] \leq 7C^{-2} |E|$. By Markov's
Inequality $P[X > 21C^{-2}|E| ] < 1/3$.  Thus, with probability at least $2/3$,
our algorithm will remove at most $21C^{-2}|E| \le 63C^{-2}n$ vertices.  Since
we have promised to remove at most $2\epsilon n$ vertices, it suffices to
select any value $C\ge \frac{8}{\sqrt{\epsilon}}$.  \end{proof}

Now, we present the proof of Lemma~\ref{lem:sparse-hard}

\begin{proof}

Suppose that the statement is false, therefore for any $\epsilon>0$ we can
solve \textsc{Sparse 3-SAT} in which the number of variables and clauses can be
upper-bounded by $N$ in expected time $2^{\epsilon N}$ using some supposed
algorithm.  Fix an arbitrary $\epsilon'>0$.  We will show how to solve an
arbitrary instance of \textsc{3-SAT} with $n$ variables and $m$ clauses in
expected time $2^{\epsilon'(n+m)}$ using this supposed algorithm for
\textsc{Sparse 3-SAT}.  If we can do this for any arbitrary $\epsilon'$, this
will contradict the randomized ETH.

Start with an arbitrary \textsc{3-SAT} instance $\phi$ with $n$ variables and
$m$ clauses. We first edit $\phi$ to ensure that each variable appears at most
three times. In particular, if $x$ appears $k>3$ times, we replace each
appearance of $x$ with a fresh variable $x_i$, $i\in [k]$, and add the clauses
$(\neg x_1\lor x_2)\land (\neg x_2\lor x_3)\land\ldots \land(\neg x_k\lor
x_1)$.  

The number of variables in the new instance is at most $n+3m$. The number of
clauses is at most $4m$. This is because every new clause and every new
variable corresponds to an occurrence of an original variable in an original
clause and there are at most $3m$ such occurrences.

We now have an instance $\phi'$ equivalent to $\phi$ with at most $n+3m$
variables and at most $4m$ clauses, such that each variable appears at most $3$
times. Let $N$ be the smallest perfect square such that $N\ge n+4m$. We have
$N<10(n+m)$. What we need now is to produce a partition of the vertices and
clauses of $\phi'$.

In order to produce this partition we invoke \cref{lem:color} on the incidence
graph of $\phi'$, that is, the bipartite graph where we have variables on one
side and clauses on the other, and edges signify that a variable appears in a
clause. Add some dummy isolated vertices on each side so that both sides of the
incidence graph contain $N$ vertices. We invoke \cref{lem:color} by setting
$\epsilon$ to be $\epsilon'/80$. We obtain a coloring of all but at most
$\frac{\epsilon'N}{40}\le \frac{\epsilon'(n+m)}{4}$ of the vertices of the
incidence graph.
 
Let $U$ be the set of variables and clauses that correspond to uncolored
vertices of the incidence graph. Then, for each such variable we produce two
formulas (one by setting it to True and one by setting it to False), and for
each such clause, at most $3$ formulas (one by setting each of the literals of
the clause to True). We thus construct at most $3^{\epsilon'(n+m)/4}\le
2^{\epsilon'(n+m)/2}$ new formulas, such that one of them is satisfiable if and
only if $\phi$ was satisfiable. We will then use the supposed algorithm for
\textsc{Sparse 3-SAT} to decide each of these formulas one by one.

Each new formula we have contains at most $N$ variables and at most $N$
clauses, and by \cref{lem:color} we have partitions of the variables and
clauses into $C \sqrt{N}$ groups, where $C$ is a constant (that depends on
$\epsilon'$). By setting $N'=\lceil C\rceil^2N$ we can view these instances as
instances of \textsc{Sparse 3-SAT}, because then the number of groups becomes
equal to the square root of the upper bound on the number of variables and
clauses, and by the properties of \cref{lem:color} there is at most one edge
between each group of variables and each group of clauses.  Since we suppose
that for all $\epsilon>0$ such instances can be solved in time $2^{\epsilon N'}$, by
setting $\epsilon =\epsilon'/50\lceil C\rceil ^2$ we can solve each formula in
$2^{\epsilon' (n+m)/5}$. The total expected running time of our algorithm is at
most $2^{\epsilon'(n+m)/2} \cdot 2^{\epsilon' (n+m)/5} \cdot (n+m)^{O(1)} \le
2^{\epsilon'(n+m)}$, so we contradict the ETH.  \end{proof}

We are now ready to prove the main theorem of this section.

\begin{theorem}\label{thm:all:slots:rand:eth}

Suppose the randomized ETH is true. Then, there exists an $\epsilon>0$ such
that {\CPD} on instances with a binary alphabet cannot be solved in time
$2^{\epsilon n^2}\cdot m^{O(1)}$.  This holds also for instances where all
slots have distinct sizes (so words cannot be reused).

\end{theorem}

\begin{proof}

Suppose for the sake of contradiction that for any fixed $\epsilon>0$, {\CPD}
on instances with a binary alphabet can be solved in time $2^{\epsilon
n^2}\cdot m^{O(1)}$.  We will then contradict \cref{lem:sparse-hard}. In
particular, we will show that for any $\epsilon'$ we can solve \textsc{Sparse
3-SAT} in time $2^{\epsilon'N}$, where $N$ is the upper bound on the number of
variables and clauses.  Fix some $\epsilon'>0$ and suppose that $\phi$ is an
instance of \textsc{Sparse 3-SAT} with at most $N$ variables and at
most $N$ clauses, where $N$ is a perfect square.  Recall that the variables are
given partitioned into $\sqrt{N}$ sets, $V_1,\ldots, V_{\sqrt{N}}$ and the
clauses partitioned into $\sqrt{N}$ sets $C_1,\ldots, C_{\sqrt{N}}$. In the
remainder, when we write $V(C_j)$ we will denote the set of variables that
appear in a clause of $C_j$.  Recall that the partition satisfies the property
that for all $i,j\in [\sqrt{N}\ ]$ we have $|V_i \cap V(C_j)| \le 1$.
Suppose that the variables of $\phi$ are ordered $x_1,x_2,\ldots, x_N$.

We construct a grid as follows: for each group $V_i$ we construct a horizontal
slot and for each group $C_j$ we construct a vertical slot, in a way that all
slots have distinct lengths. More precisely, the $i$-th horizontal slot, for
$i\in [\sqrt{N}]$ is placed on row $2i-1$, starts in the first column and has
length $2\sqrt{N}+2i$. The $j$-th vertical slot is placed in column $2j-1$,
starts in the first row and has length $5\sqrt{N}+2j$. (As usual, we number the
rows and columns top-to-bottom and left-to-right). Observe that all horizontal
slots intersect all vertical slots; in particular, the cell in row $2i-1$ and
column $2j-1$ is shared between the $i$-th horizontal and $j$-th vertical slot,
for $i,j\in [\sqrt{N}]$.  We define $\L$ to contain two letters $\{0,1\}$.

What remains is to describe the dictionary.

\begin{itemize}

    \item For each $i\in [\sqrt{N} ]$ and for each assignment function
$\sigma:V_i\to \{0,1\}$ we construct a word $w_\sigma$ of length
$2\sqrt{N}+2i$. The word $w_\sigma$ has the letter $0$ in all positions, except
positions $2j-1$, for $j\in [\sqrt{N}]$. For each such $j$, we consider
$\sigma$ restricted to $V_i\cap V(C_j)$. By the properties of \textsc{Sparse
3-SAT}, we have $|V_i\cap V(C_j)|\le 1$. If $V_i\cap V(C_j)=\emptyset$ then we
place letter $0$ in position $2j-1$; otherwise we set in position $2j-1$ the
letter that corresponds to the value assigned by $\sigma$ to the unique
variable of $V_i\cap V(C_j)$.
    
\item For each $j\in [\sqrt{N} ]$ and for each \emph{satisfying} assignment
function $\sigma:V(C_j)\to \{0,1\}$, that is, every assignment function that
satisfies all clauses of $C_j$, we construct a word $w'_\sigma$ of length
$5\sqrt{N}+2j$. The word $w'_\sigma$ has the letter $0$ in all positions,
except positions $2i-1$, for $i\in [\sqrt{N}]$. For each such $i$, we consider
$\sigma$ restricted to $V_i\cap V(C_j)$. If $V_i\cap V(C_j)=\emptyset$ then we
place letter $0$ in position $2i-1$; otherwise we set in position $2i-1$ the
letter that corresponds to the value assigned by $\sigma$ to the unique
variable of $V_i\cap V(C_j)$.

\end{itemize}

The construction is now complete. We claim that if $\phi$ is satisfiable, then
it is possible to fill out the grid we have constructed. Indeed, fix a
satisfying assignment $\sigma$ to the variables of $\phi$. For each $i\in
[\sqrt{N}]$ let $\sigma_i$ be the restriction of $\sigma$ to $V_i$. We place in
the $i$-th horizontal slot the word $w_{\sigma_i}$. Similarly, for each $j\in
[\sqrt{N}]$ we let $\sigma'_j$ be the restriction of $\sigma$ to $V(C_j)$ and
place $w'_{\sigma'_j}$ in the $j$-th vertical slot. Now if we examine the cell
shared by the $i$-th horizontal and $j$-th vertical slot, we can see that it
contains a letter that represents $\sigma$ restricted to (the unique variable
of) $V_i\cap V(C_j)$ or $0$ if $V_i\cap V(C_j)=\emptyset$, and both the
horizontal and vertical word place the same letter in that cell.

For the converse direction, if the grid is filled, we can extract an assignment
$\sigma$ for the variables of $\phi$ as follows: for each $x\in V_i$ we find a
$C_j$ such that $x$ appears in some clause of $C_j$ (we can assume that every
variable appears in some clause). We then look at the cell shared between the
$i$-th horizontal and the $j$-th vertical slot. The letter we have placed in
that cell gives an assignment for the variable contained $V_i\cap V(C_j)$, that
is $x$.  Having extracted an assignment to all the variables, we claim it must
satisfy $\phi$. If not, there is a group $C_j$ that contains an unsatisfied
clause.  Nevertheless, in the $j$-th vertical slot we have placed a word that
corresponds to a \emph{satisfying} assignment for the clauses of $C_j$, call it
$\sigma_j$. Then $\sigma_j$ must disagree with $\sigma$ in a variable $x$ that
appears in $C_j$. Suppose this variable is part of $V_i$.  Then, this would
contradict the fact that we extracted an assignment for $x$ from the word
placed in the $i$-th horizontal slot.

Observe that the new instance has $n=2\sqrt{N}$ slots. If there exists an
algorithm that solves {\CPD} in time $2^{\epsilon n^2} m^{O(1)}$ for any
$\epsilon>0$, we set $\epsilon = \epsilon'/8$ (so $\epsilon$ only depends on
$\epsilon'$) and execute this algorithm on the constructed instance.  We
observe that $m\le 2\sqrt{N}\cdot 7^{\sqrt{N}}$, and that $2^{\epsilon n^2} \le
2^{\epsilon'N/2}$.  Assuming that $N$ is sufficiently large, using the supposed
algorithm for {\CPD} we obtain an algorithm for \textsc{Sparse 3-SAT} with
complexity at most $2^{\epsilon' N}$.  Since we can do this for arbitrary
$\epsilon'$, this contradicts the randomized ETH. \end{proof}

\section{Approximability of {\CPO}} \label{sec:approximability}


This section begins with a $\big(\frac{1}{2}+O(\frac{1}{n})\big)$-approximation algorithm which works when words can, or cannot, be reused. After that, we prove that under the unique games conjecture, an approximation algorithm with a significantly better ratio  is unlikely.

\begin{theorem}\label{cp_opt_is_1/2+approx} {\CPO} is $(\frac{1}{2}+\frac{1}{2(\varepsilon n+1)})$-approximable in polynomial time, for all $\varepsilon \in(0,1]$.
\end{theorem}

\begin{proof}
Fix some  $\varepsilon \in (0,1]$. Let $k_v:=\min (\lceil \frac{1}{\varepsilon} \rceil,n-h)$ and $r_v:=\lceil \frac{n-h}{k_v} \rceil$,
where $h$ is the number of horizontal slots in the grid. 
Create $r_v$ groups of vertical slots $G_1, \ldots , G_{r_v}$ such that $|G_i| \le k_v$ for all $i \in [r_v]$ and $G_1 \cup \ldots \cup G_{r_v}$ covers the entire set of vertical slots. 
For each $G_i$, guess an optimal choice of words, i.e., identical to a global optimum, and complete this partial solution by filling the 
horizontal slots (use the aforementioned matching technique where the words selected for $G_i$ are excluded from $\D$). Each slot of $\bigcup_{j\neq i} G_j$ gets the empty word.

Since $|G_i| \le k_v$, guessing an optimal choice of words for $G_i$ by brute force requires at most $m^{k_v}$ combinations.  This is done $r_v$ times (once for each $G_i$). The maximum matching runs in time ${\cal O} ((m+n)^2 \cdot m  n)$. In all, the time complexity of the algorithm is ${\cal O} (m^{k_v} \cdot {r_v} \cdot (m+n)^2 \cdot m  n) \le  {\cal O} (m^{1/\varepsilon} \cdot \varepsilon n \cdot (m+n)^2 \cdot m  n)$.

Assume that, given an optimal solution, $W^*_H$ and $W^*_V$ are the total weight of the words assigned to the horizontal and vertical slots, respectively, both including the shared cells. Furthermore, let $W^*_S$ be the weight of the letters assigned to the shared cells in the optimal solution. Observe that the weight of the optimal solution is $W^*_H+ W^*_V -W^*_S$ and the weight of our solution is at least $W^*_H+\frac{1}{r_v}(W^*_V-W^*_S)$.



We repeat the same process, but the roles of vertical and horizontal slots are interchanged. 
Fix  a parameter $k_h:=\min (\lceil \frac{1}{\varepsilon} \rceil,h)$. 
Create $r_h:=\lceil \frac{h}{k_h} \rceil$ groups of horizontal slots $G_1, \ldots , G_{r_h}$ such that
$|G_i| \le k_h$ for all $i \in [r_h]$ and $G_1 \cup \ldots \cup G_{r_h}$ covers the entire set of horizontal slots. For each $G_i$, guess an optimal choice of words and complete this partial solution 
by filling the vertical slots. Each slot of $\bigcup_{j\neq i} G_j$ gets the empty word.  

Using the same arguments as above, we can conclude that the time complexity is $ { O} (m^{1/\varepsilon} \cdot \varepsilon n \cdot (m+n)^2 \cdot m  n)$ and that we return a solution of weight at least $W^*_V+\frac{1}{r_h} (W^*_H- W^*_S)$.

Finally, between the two solutions, we return the one with the greater weight. It remains to argue about the approximation ratio. We need to consider two cases: $W^*_H \ge W^*_V$ and $W^*_V > W^*_H$.

Suppose $W^*_H \ge W^*_V$. The first approximate solution has value $W^*_H+\frac{1}{r_v} (W^*_V- W^*_S) \ge \frac{1+1/r_v}{2}(W^*_H + W^*_V- W^*_S)$. If $k_v=n-h$  then $r_v=1$ and our approximation ratio is $1$. Otherwise, $k_v=\lceil \frac{1}{\varepsilon} \rceil$ and  $r_v=\lceil \frac{n-h}{\lceil {1}/{\varepsilon} \rceil} \rceil \le 
\frac{n-h}{\lceil {1}/{\varepsilon} \rceil} +1=\frac{n-h+\lceil {1}/{\varepsilon} \rceil}{\lceil {1}/{\varepsilon} \rceil}$. It follows that $\frac{1}{r_v} \ge 
\frac{\lceil {1}/{\varepsilon} \rceil}{n-h+\lceil {1}/{\varepsilon} \rceil}$. Use $n-h+\lceil {1}/{\varepsilon} \rceil \le n+\frac{1}{\varepsilon}$ and $\lceil {1}/{\varepsilon} \rceil \ge {1}/{\varepsilon}$ to get that $\frac{1}{r_v} \ge \frac{{1}/{\varepsilon}}{n+{1}/{\varepsilon}}=\frac{1}{\varepsilon n+1}$. Our approximation ratio is
at least $\frac{1+1/(\varepsilon n+1)}{2}$.    

Suppose $W^*_V > W^*_H$. The second approximate solution has value $W^*_V+\frac{1}{r_h} (W^*_H- W^*_S) > \frac{1+1/r_h}{2}(W^*_H + W^*_V- W^*_S)$. If $k_h=h$, then our approximation ratio is $1$. Otherwise, $k_h=\lceil \frac{1}{\varepsilon} \rceil$ and, using the same arguments, our approximation ratio is at least $\frac{1+1/(\varepsilon n+1)}{2}$.

Note that $\frac{1+1/(\varepsilon n+1)}{2} \le 1$. 
In all, we have a $\frac{1+1/(\varepsilon n+1)}{2}$-approximate solution in ${\cal O} (m^{1/\varepsilon} \cdot \varepsilon n \cdot (m+n)^2 \cdot m  n)$ for all $\varepsilon \in (0,1]$. \end{proof}

The previous approximation algorithm 
only achieves an
approximation ratio of $\frac{1}{2}+O(\frac{1}{n})$, which tends to
$\frac{1}{2}$ as $n$ increases. At first glance this is quite disappointing, as
someone can observe that a ratio of $\frac{1}{2}$ is achievable simply by
placing words only on the horizontal or the vertical slots of the instance. 
Nevertheless, 
we are going to show that this performance is justified, as
improving upon this trivial approximation ratio would falsify the Unique Games
Conjecture (UGC).

Before we proceed, let us recall some relevant definitions regarding Unique
Games. The \textsc{Unique Label Cover} problem is defined as follows: we are
given a graph $G=(V,E)$, with some arbitrary total ordering $\prec$ of $V$, an integer $R$, and for each $(u,v)\in E$ with $u\prec v$ a 1-to-1
constraint $\pi_{(u,v)}$ which can be seen as a permutation on $[R]$. 
The
vertices of $G$ are considered as variables of a constraint satisfaction
problem, which take values in $[R]$. Each constraint $\pi_{(u,v)}$ defines for
each value of $u$ a unique value that must be given to $v$ in order to satisfy
the constraint. The goal is to find an assignment to the variables that
satisfies as many constraints as possible. The Unique Games Conjecture states
that for all $\epsilon>0$, there exists $R$, such that distinguishing instances
of \textsc{Unique Label Cover} for which it is possible to satisfy a
$(1-\epsilon)$-fraction of the constraints from instances where no assignment
satisfies more than an $\epsilon$-fraction of the constraints is NP-hard. In
this section we will need a slightly different version of this conjecture,
which was defined by Khot and Regev as the Strong Unique Games Conjecture.
Despite the name, Khot and Regev showed that this version is implied by the
standard UGC. The precise formulation is the following:

\begin{theorem}\label{thm:ugc}[Theorem 3.2 of \cite{KhotR08}]  If the 
Unique Games Conjecture 
is true, then for all $\epsilon>0$ it is NP-hard to
distinguish between the following two cases of instances of \textsc{Unique
Label Cover} $G=(V,E)$:
\begin{itemize}
\item (Yes case): There exists a set $V'\subseteq V$ with $|V'|\ge
(1-\epsilon)|V|$ and an assignment for $V'$ such that all constraints with both
endpoints in $V'$ are satisfied.

\item (No case): For any assignment to $V$, for any set $V'\subseteq V$ with
$|V'|\ge \epsilon |V|$, there exists a constraint with both endpoints in $V'$
that is violated by the assignment.

\end{itemize}
\end{theorem}

Using the version of the UGC given in Theorem \ref{thm:ugc} we are ready to
present our hardness of approximation argument for the crossword puzzle.

\begin{theorem} \label{thm:UGC}
Suppose that the Unique Games Conjecture is true. Then, for all $\epsilon$ with
$\frac{1}{4}>\epsilon>0$, there exists an alphabet $\Sigma_\epsilon$ such that
it is NP-hard to distinguish between the following two cases of instances of
the crossword problem on alphabet $\Sigma_\epsilon$:

\begin{itemize}

\item (Yes case): There exists a valid solution that fills a
$(1-\epsilon)$-fraction of all cells.

\item (No case): No valid solution can fill more than a
$(\frac{1}{2}+\epsilon)$-fraction of all cells.

\end{itemize}

Moreover, the above still holds if all slots have distinct lengths (and hence
reusing words is trivially impossible).

\end{theorem}

\begin{proof}

Fix an $\epsilon>0$. We will later define an appropriately chosen value
$\epsilon'\in (0,\epsilon)$ whose value only depends on $\epsilon$. We present
a reduction from a \textsc{Unique Label Cover} instance, as described in
Theorem \ref{thm:ugc}.  In particular, suppose we have an instance $G=(V,E)$,
with $|V|=n$, alphabet $[R]$, such that (under UGC) it is NP-hard to
distinguish if there exists a set $V'$ of size $(1-\epsilon')n$ that satisfies
all its induced constraints, or if all sets $V'$ of size $\epsilon'n$ induce at
least one violated constraint for any assignment. Throughout this proof we
assume that $n$ is sufficiently large (otherwise the initial instance is easy).
In particular, let $n>\frac{20}{\epsilon}$.

We construct an instance of the crossword puzzle that fits in an $N\times N$
square, where $N=4n+n^2$. We number the rows $1,\ldots, N$ from top to bottom
and the columns $1,\ldots,N$ from left to right. The instance contains $n$
horizontal and $n$ vertical slots. For $i\in [n]$, the $i$-th horizontal slot
is placed in row $2i$, starting at column $1$, and has length $2n+n^2+i$. For
$j\in [n]$, the $j$-th vertical slot is placed in column $2j$, starts at row
$1$ and has length $3n+n^2+j$. 
Observe that all horizontal slots intersect all
vertical slots and in particular, for all $i,j\in [n]$ the cell in row $2i$,
column $2j$ belongs to the $i$-th horizontal slot and the $j$-th vertical slot.
Furthermore, each slot has a distinct length, as the longest horizontal slot
has length $3n+n^2$ while the shortest vertical slot has length $3n+n^2+1$.

We define the alphabet as $\Sigma_\epsilon = [R]\cup \{*\}$.
Before we define our dictionary, let us give some intuition. Let
$V=\{v_1,\ldots, v_n\}$. The idea is that a variable $v_i\in V$ of the original
instance will be represented by both the $i$-th horizontal slot and the $i$-th
vertical slot. In particular, we will define, for each $\alpha\in [R]$ a pair of words 
that we can place in these slots to represent the fact that $v_i$ is assigned with the
value $\alpha$. We will then ensure that if we place words on both the $i$-th
horizontal slot and the $j$-th horizontal slot, where $(v_i,v_j)\in E$, then the 
assignment that can be extracted by reading these words will satisfy the
constraint $\pi_{(v_i,v_j)}$. The extra letter $*$  represents an
indifferent assignment (which we need if $(v_i,v_j)\not\in E$).

Armed with this intuition, let us define our dictionary. 

\begin{itemize}

\item For each $i\in [n]$, for each $\alpha\in [R]$ we define a word
$\w_{(i,\alpha)}$ of length $2n+n^2+i$. The word $\w_{(i,\alpha)}$ has the
character $*$ everywhere except at position $2i$ and at positions $2j$ for
$j\in [n]$ and $(v_i,v_j)\in E$. In these positions the word $\w_{(i,\alpha)}$
has the character $\alpha$.

\item For each $j\in [n]$, for each $\alpha\in [R]$ we define a word
$\w'_{(j,\alpha)}$ of length $3n+n^2+j$. The word $\w'_{(j,\alpha)}$ has the
character $*$ everywhere except at position $2j$ and at positions $2i$ for
$i\in [n]$ and $(v_i,v_j)\in E$. In position $2j$ we have the character
$\alpha$. In position $2i$ with $(v_i,v_j)\in E$, we place the character
$\beta\in [R]$ such that the constraint $\pi_{(v_i,v_j)}$ is satisfied by
assigning $\beta$ to $v_i$ and $\alpha$ to $v_j$. (Note that $\beta$ always
exists and is unique, as the constraints are permutations on $[R]$, that is,
for each value $\alpha$ of $v_j$ there exists a unique value $\beta$ of $v_i$
that satisfies the constraint).

\end{itemize}

This completes the construction. Suppose now that $V=\{v_1,\ldots,v_n\}$ and
that we started from the Yes case of \textsc{Unique Label Cover}, that is,
there exists a set $V'\subseteq V$ such that $|V'|\ge (1-\epsilon')n$ and all
constraints induced by $V'$ can be simultaneously satisfied. Fix an assignment
$\sigma : V'\to [R]$ that satisfies all constraints induced by $V'$. For each
$i\in [n]$ such that $v_i\in V'$ we place in the $i$-th horizontal slot (that
is, in row $2i$) the word $\w_{(i,\sigma(v_i))}$.  For each $j\in [n]$ such that
$v_j\in V'$ we place in the $j$-th vertical slot the word
$\w'_{(j,\sigma(v_j))}$. We leave all other slots empty. We claim that this
solution is valid, that is, no shared cell is given different values from its
horizontal and vertical slot. To see this, examine the cell in row $2i$ and
column $2j$. If both of the slots that contain it are filled, then $v_i, v_j\in
V'$. If $(v_i,v_j)\not\in E$ and $i\neq j$, then the cell contains $*$ from
both words. If $i=j$, then the cell contains $\sigma(v_i)$ from both words. If
$i\neq j$ and $(v_i,v_j)\in E$, then the cell contains $\sigma(v_i)$. This is
consistent with the vertical word, as the constraint $\pi_{(v_i,v_j)}$ is
assumed to be satisfied by $\sigma$. We now observe that this solution covers
at least $2(1-\epsilon')n^3$ cells, as we have placed $2(1-\epsilon')n$ words,
each of length at least $n^2+2n$, that do not pairwise intersect beyond their
first $2n$ characters.

Suppose now we started our construction from a No instance of
\textsc{Unique Label Cover}. We claim that the optimal solution in the new
instance cannot cover significantly more than half the cells. In particular,
suppose a solution covers at least $(1+\epsilon')n^3+10n^2$ cells. We
claim that the solution must have placed at least $(1+\epsilon')n$ words.
Indeed, if we place at most $(1+\epsilon')n$ words, as the longest word has
length $n^2+4n$, the maximum number of cells we can cover is
$(1+\epsilon')n(n^2+4n) \le
(1+\epsilon')n^3+4(1+\epsilon')n^2<(1+\epsilon')n^3+10n^2$. Let $x$ be the
number of indices $i\in [n]$ such that the supposed solution has placed a word
in both the $i$-th horizontal slot and the $i$-th vertical slot. We claim
that $x\ge \epsilon' n$. Indeed, if $x<\epsilon' n$, then the total number of
words we might have placed is at most $(n-x)+2x < (1+\epsilon')n$, which
contradicts our previous observation that we placed at least
$(1+\epsilon')n$ words. Let $V'\subseteq V$ be defined as the set of $v_i\in V$
such that the solution places words in the $i$-th horizontal and vertical slot.
Then $|V'|\ge \epsilon' n$. We claim that it is possible to satisfy all the
constraints induced by $V'$ in the original instance, obtaining a
contradiction. Indeed, we can extract an assignment for each $v_i\in V'$ by
assigning to $v_i$ value $\alpha$ if the $i$-th horizontal slot contains the
word $\w_{(i,\alpha)}$. Note that the $i$-th horizontal slot must contain such a
word, as these words are the only ones that have an appropriate length. Observe
that in this case the $i$-th vertical slot must also contain $\w'_{(i,\alpha)}$.
Now, for $v_i,v_j\in V'$, with $(v_i,v_j)\in E$ we see that $\pi_{(v_i,v_j)}$
is satisfied by our assignment, otherwise we would have a conflict in the cell
in position $(2i,2j)$. Therefore, in the No case, it must be impossible
to fill more than $(1+\epsilon')n^3+10n^2$ cells.

The only thing that remains is to define $\epsilon'$. Let $C$ be the total
number of cells in the instance. Recall that we proved that in the Yes case we
cover at least $2(1-\epsilon')n^3$ cells and in the No case at most
$(1+\epsilon')n^3+10n^2$ cells. So we need to define $\epsilon'$ such that
$2(1-\epsilon')n^3\ge (1-\epsilon)C$ and $(1+\epsilon')n^3+10n^2\le
(\frac{1}{2}+\epsilon)C$. To avoid tedious calculations, we observe that
$2n^3\le C \le 2n^3+8n^2$.  Therefore, it suffices to have $2(1-\epsilon')n^3
\ge 2(1-\epsilon)(n^3+4n^2)$ and $(1+\epsilon')n^3+10n^2 \le (1+2\epsilon)
n^3$. The first inequality is equivalent to $(\epsilon-\epsilon')n\ge
4(1-\epsilon)$ and the second inequality is equivalent to
$(2\epsilon-\epsilon')n\ge 10 $. Since we have assumed that $n\ge 20/\epsilon$,
it is sufficient to set $\epsilon'=\epsilon/2$.  \end{proof}

\section{Conclusion}

We studied the parameterized complexity 
of some crossword puzzles
under several different parameters and we gave some positive results followed by proofs which show that our algorithms are essentially optimal. 
Based on our results the most natural questions that arise are:  What is the complexity of {\CPD} 
when the grid graph is a matching and the alphabet has size $2$? Can Theorem~\ref{thm:all:slots:rand:eth} be strengthened by starting from ETH 
instead of randomized ETH? Can we beat the $1/2$ approximation ratio of {\CPO} if we restrict our instances? 
Can Theorem~\ref{thm:ugc} be strengthened by dropping the UGC?  Furthermore, it would be interesting to investigate if there exist non trivial instances of the problem that can be solved in polynomial time. 
Finally, we could consider a variation of the crossword puzzle problems where each word can be used a 
given number of times. This would be an intermediate case between word reuse and no word reuse.   



\bibliography{references}

\end{document}